\documentclass{revtex4}

\usepackage{amsmath,amssymb,amsthm}
\usepackage{graphicx}

\newtheorem{theorem}{Theorem}

\newcommand\T{\rule{0pt}{5ex}}
\newcommand\B{\rule[-3ex]{0pt}{0pt}}
\newcommand\Ti{\rule{0pt}{3ex}}
\newcommand\Bi{\rule[-2ex]{0pt}{0pt}}

\begin{document}

\title{Relative Fisher information of discrete classical orthogonal polynomials}

\author{J.S. Dehesa$^{\rm a,c}$ $^{\ast}$ \thanks{$^\ast$Email: dehesa@ugr.es}, P. S\'anchez-Moreno$^{\rm b,c}$$^{\dagger}$\thanks{$^\dagger$Corresponding author. Email: pablos@ugr.es} \& R.J. Y\'a\~nez$^{\rm b,c}$$^{\ddagger}$\thanks{$^\ddagger$Email: ryanez@ugr.es}\\
\vspace{6pt}  $^{\rm a}${\em{Department of Atomic, Molecular and Nuclear Physics, University of Granada, Spain}};\\
$^{\rm b}${\em{Department of Applied Mathematics, University of Granada, Spain}};\\
$^{\rm c}${\em{Institute ``Carlos I'' for Theoretical and Computational Physics, University of Granada, Spain.}}}

\begin{abstract}
The analytic information theory of discrete distributions was initiated in 1998 by C. Knessl, P. Jacquet and S. Szpankowski who addressed the precise evaluation of the Renyi and Shannon entropies of the Poisson, Pascal (or negative binomial) and binomial distributions. They were able to derive various asymptotic approximations and, at times, lower and upper bounds for these quantities. Here we extend these investigations in a twofold way. First, we consider a much larger class of distributions, the Rakhmanov distributions $\rho_n(x)=\omega(x)y_n^2(x)$, where $\{y_n(x)\}$ denote the sequences of discrete hypergeometric-type polynomials which are orthogonal with respect to the weight function $\omega(x)$ of Poisson, Pascal, binomial and hypergeometric types; that is the polynomials of Charlier, Meixner, Kravchuk and Hahn. Second, we obtain the explicit expressions for the relative Fisher information of these four families of Rakhmanov distributions with respect to their respective weight functions.
\end{abstract}


\maketitle

\section{Introduction}

The knowledge of the algebraic and spectral properties of the orthogonal polynomials in a discrete variable \cite{alvareznodarse_03,chihara_78,ismail_05,koekoek_98,nikiforov_91} as well as the elucidation of their universal structure \cite{odake:jmp08} have been issues of permanent interest since the early years of the last century up until now, not only because of its mathematical interest \cite{alvareznodarse:amc02,aptekarev:ca08,alvareznodarse:jcam97,chihara_78,doliwa:jpa07,garcia:jcam95,lee_07,koekoek_98,nikiforov_91,ismail_05,odake:jmp08,lorente:jpa01} but also because of the increasing number of applications of these functions in so many scientific and technological fields \cite{boykin:ejp05,doliwa:jpa07,lorente:pla01,lorente:jpa01,lorente:jcam03,mickens:jdea05,smirnov:jpa84,deuflhard:icse89,vercin:jmp98,meiler:08,dehesa:maa97}. In particular, the classical or hypergeometric discrete orthogonal polynomials do not only play a relevant role in the theory of difference analogues of special functions and other branches of mathematics \cite{nikiforov_91,smirnov:jpa84,alvareznodarse:etna07,buyarov:jat99,dehesa:maa97,doliwa:jpa07,lorente:jpa01}, but also for mathematical modelling of a great deal of simple \cite{atakishiyev:tmp91,boykin:ejp05,lorente:pla01,lorente:jcam03,mickens:jdea05,smirnov:jpa84,suslov:sjnp84,vercin:jmp98} and complex \cite{deuflhard:icse89,carballo:aml01,savva:itsf00,defazio:ijqc03,meiler:08,deuflhard:icse89} systems, as well as for the compression of information for signal processing \cite{jacquet_98,nikiforov_88,odake:jmp08}.

From the 1990's it emerges an information theory of the special functions of applied mathematics and mathematical physics, which allows to compute the information-theoretic properties of the solvable and quasi-solvable quantum-mechanical systems.  This theory has been primarily devoted to orthogonal polynomials in a continuous variable with standard (i.e. non varying) weights \cite{aptekarev:rassm95,dehesa:maa97,vanassche:jmp95,dehesa:jcam01,buyarov:sjsc04,dehesa:jcam06,sanchezruiz:jcam05,yanez:jmp08,aptekarev:jcam09}, but soon was extended to varying weights \cite{buyarov:jat99} and to other special functions with non-polynomic form \cite{dehesa:jmp07,yanez:jmp99,sanchezmoreno:jpa05,dehesa:ijbc02,dehesa:jmp03}. Presently this theory is being extended to the orthogonal polynomials in a discrete variable, being the only results known those of Larsson-Cohn \cite{larsson:jat02} and Aptekarev et al. \cite{aptekarev:ca08} about the asymptotics of the Shannon information entropy of Charlier polynomials, and Dehesa et al. \cite{dehesa_05} about the standard or mean-square-root deviation of the four classical families and some bounds on the Shannon entropic power and the Fisher information of these families.

The information-theoretic/spreading properties of the discrete polynomials $P_n(x)$, $x\in[a,b]$, orthogonal with respect to a certain weight function $\omega(x)$ are given by the corresponding information-theoretic measures of the associated Rakhmanov probability density
\begin{equation}
\rho_n(x)=\frac{1}{d_n^2}[P_n(x)]^2 \omega(x),
\label{eq:rhon}
\end{equation}
where the normalization constant is given by the orthogonality relation
\begin{equation}
\sum_{x=a}^{b-1}\omega(x) P_n(x) P_m(x)=\delta_{n,m} d_n^2.
\label{eq:orthogonality}
\end{equation}
This density is important from both mathematical and physical standpoints. It governs the behaviour of the ratio $P_{n+1}(x)/P_n(x)$ when $n$ goes to infinity as Rakhmanov showed in 1977 \cite{rakhmanov:mus77}, and it characterizes the quantum-mechanical probability densities of ground and excited states of numerous physical systems (see, e.g. \cite{dehesa:jcam01,nikiforov_91,nikiforov_88,vercin:jmp98}).

Beyond the variance, the R\'enyi and Shannon information entropies have been used to quantify the spreading of various simple discrete distributions. Particularly relevant for these quantities are the integral representations of C. Knessl \cite{knessl:aml98} and the depoissonization process of P. Jacquet  and S. Szpankowski \cite{jacquet_98,jacquet:itit99} used to obtain the full asymptotics of the R\'enyi and Shannon entropies of the discrete Poisson, negative binomial or Pascal and binomial  distributions. The explicit computation of these quantities is still today a formidable open task. In this work we want to extend the results of these authors in a two-fold way. First, by considering a much larger class of discrete distributions composed by the Rakhmanov distributions associated to the classical discrete polynomials $P_n(x)$ which are defined by Eq. (\ref{eq:rhon}), where the weight function $\omega(x)$ has a Poisson, Pascal, binomial and hypergeometric analytical form; that is, the class of the Rakhmanov probability densities associated to the Charlier, Meixner, Kravchuk and Hahn polynomials. Second, by obtaining the explicit expression of the relative Fisher information of the four families of these Rakhmanov distributions. This information-theoretic quantity, unlike the R\'enyi and Shannon entropies which have a global character, has a locality property because it is a gradient functional of the density. So, it provides a quantifier of the gradient content of the density. For unimodal densities, it increases as the density is more and more localized. In general, for oscillating densities the Fisher information gives a measure of the oscillatory character of the density under consideration.

Although the notion of Fisher information was initially introduced by  Ronald A. Fisher as a measure to estimate a parameter of a probability density, we shall follow B. Roy Frieden who realized that the locality Fisher information (also called intrinsic accuracy) for continuous distributions is much more useful for applications in science and technology \cite{frieden_04}. For completeness we note that the parameter-based Fisher information was calculated for various continuous classical real orthogonal polynomials \cite{dehesa:jcam07} as well as for the discrete Hahn polynomials \cite{dominici:jcam09}, while the locality or shift-invariant Fisher information has been exactly calculated for the Hermite, Laguerre and Jacobi polynomials. The discretization of the locality Fisher information is discussed in detail in \cite{sanchezmoreno_08}, where the definition of the Fisher information for discrete orthogonal polynomials is introduced.
On the other hand the relative Fisher information of orthogonal polynomials $\{P_n(x)\}$ with respect to its weight function $\omega(x)$, $x\in(a,b)$, has been recently defined \cite{yanez:jmp08} in the continuous case as
\begin{equation}
I_\omega[P_n]=\int_a^b\rho_n(x)\left[\frac{d}{dx}\ln\frac{\rho_n(x)}{\omega(x)}\right]^2dx
=4\int_a^b \omega(x) \left[\frac{d}{dx}P_n(x)\right]^2dx.
\label{eq:fishercontinuous}
\end{equation}
This quantity has been discussed in detail and explicitly calculated for the Laguerre, Hermite and Jacobi polynomials and other special functions \cite{yanez:jmp08}. In the present work we define the relative Fisher information of the discrete orthogonal polynomials \cite{sanchezmoreno_08} as the discretization of the continuous one given by Eq. (\ref{eq:fishercontinuous}), that is by
\begin{equation}
I_\omega[P_n]=\frac{1}{d_n^2}\sum_{x=a}^{b-1} \omega(x)[\Delta P_n(x)]^2,
\label{eq:fisherdefinition}
\end{equation}
which will be called by as the relative Fisher information of the discrete polynomial $P_n(x)$, where $\Delta P_n(x)=P_n(x+1)-P_n(x)$ is the forward difference operator. Remark that $I_\omega[P_n]$ is the relative Fisher information of the Rakhmanov probability density given by (\ref{eq:rhon}) with respect to the weight function $\omega(x)$ of the polynomials. Moreover, it is always non-negative, being zero if and only if $n = 0$. So, this quantity measures the separation between the Rakhmanov density and the weight function of the polynomials or the perturbation experienced in the weight function by the action of the polynomials themselves.

Here we obtain the explicit expressions of this quantity for all the four families of discrete classical orthogonal polynomials. The structure of the paper is the following. In Section \ref{sec2} we describe two different methods to compute the relative Fisher information of the discrete polynomials: one based on the ladder relation satisfied by these polynomials and the other one makes use of their second order hypergeometric difference equation. Then, in Section \ref{sec3} we apply the ladder-relation-based method to calculate the relative Fisher information of the Charlier, Meixner, Kravchuk and Hahn polynomials. Finally, the resulting expressions are numerically investigated and discussed.

\section{Methodology}
\label{sec2}

In this Section we describe two methods to calculate the relative Fisher information $I_\omega[P_n]$ given by (\ref{eq:fisherdefinition}) for the discrete classical orthogonal polynomials  $P_n(x)$. These methods make use of two different characterizations \cite{alvareznodarse_03,chihara_78,ismail_05,nikiforov_91} of these discrete special functions: the ladder relation (also called difference or structure formula) and the second order hypergeometric difference equation.

The ladder-relation-based method first makes use of the difference formula of $P_n(x)$ which produces another polynomial of the same type of degree $n-1$ and parameters shifted in one unity save for the Charlier case (where the parameter remains unaltered), and then a connection formula which allows one to expand this polynomial in terms of the polynomials $P_j(x)$ with $j=0,1,\ldots,n-1$, so that we finally have the expansion
\begin{equation}
\Delta P_n(x)=\sum_{j=0}^{n-1} a_{j,n} P_j(x).
\label{eq:linearization}
\end{equation}
Then, the substitution of Eq. (\ref{eq:linearization}) into Eq. (\ref{eq:fisherdefinition}) together with the orthogonality relation (\ref{eq:orthogonality}) yields the following expression 
\[
I_\omega[P_n]=\frac{1}{d_n^2}\sum_{x=a}^{b-1} \omega(x)[\Delta P_n(x)]^2=\frac{1}{d_n^2}\sum_{j=0}^{n-1}a_{j,n} d_j^2.
\]
for the relative Fisher information of the polynomial $P_n(x)$. It only remains, of course, to use the known normalization constant $d_n^2$ and the expansion coefficients $a_{j,n}$ which will be computed for the four classical families (Charlier, Meixner, Kravchuk and Hahn). Full details for each family together with the final result are given in the next Section.

The difference-equation-based method to calculate $I_\omega[P_n]$ is based on the following result.

\begin{theorem}
Let $\rho_n(x)$ be the Rakhmanov probability density (\ref{eq:rhon}) associated to the discrete polynomial $P_n(x)$ which satisfies the orthogonality relation (\ref{eq:orthogonality}) and the difference equation
\begin{equation}
\sigma(x)\Delta\nabla P_n(x) +\tau(x)\Delta P_n(x)+\lambda P_n(x)=0,
\label{eq:differenceequation}
\end{equation}
where $\sigma(x)$ and $\tau(x)$ are polynomials of at most second and first degree, respectively, and $\lambda$ is a constant. The difference operators are defined as $\Delta P_n(x)=P_n(x+1)-P_n(x)$ and $\nabla P_n(x)=P_n(x)-P_n(x-1)$. Let $I_\omega[P_n]$ be the relative Fisher information of $P_n(x)$ given by Eq. (\ref{eq:fisherdefinition}). Then,
\begin{equation}
I_\omega[P_n]=\frac{1}{d_n^2}\left( \left.\omega(x-1)P_n^2(x)\right|_a^b+\left\langle \frac{\sigma(x)}{\tau(x-1)+\sigma(x-1)}\right\rangle \right)-1,
\label{eq:fishersigmatau}
\end{equation}
with the expectation value $\langle f(x)\rangle$ defined as
\[
\langle f(x)\rangle=\sum_{x=a}^{b-1}\rho_n(x) f(x).
\]
\end{theorem}

\begin{proof}
Using the identity
\[
 [\Delta P_n(x)]^2=\Delta P_n^2(x)-2P_n(x)\Delta P_n(x).
\]
one has that the summation involved in the relative Fisher information (\ref{eq:fisherdefinition}) can be expressed as
\[
\sum_{x=a}^{b-1} \omega(x)[\Delta P_n(x)]^2=\sum_{x=a}^{b-1} \omega(x) \Delta P_n^2(x)-2\sum_{x=a}^{b-1}\omega(x) P_n(x)
\Delta P_n(x).
\]
Taking into account that $\Delta P_n(x)$ is a polynomial of degree $n-1$ on $x$, the second sum vanishes. Thus,
\[
\sum_{x=a}^{b-1} \omega(x)[\Delta P_n(x)]^2= \sum_{x=a}^{b-1} \omega(x) \Delta P_n^2(x) = \sum_{x=a}^{b-1} \omega(x) P_n^2(x+1) -\sum_{x=a}^{b-1} \omega(x) P_n^2(x),
\]
where we have used the definition of the forward operator in the second equality. Now, one realizes that the second sum is just the normalization constant because of Eq. (\ref{eq:orthogonality}). Moreover, for the first sum we use the summation by parts
\[
\sum_{x=a}^{b-1}f(x+1)=f(x)|_a^b+\sum_{x=a}^{b-1}f(x),
\]
so that one finds that
\[
\sum_{x=a}^{b-1} \omega(x) \Delta P_n^2(x)=\left.\omega(x-1)P_n^2(x)\right|_a^b+\sum_{x=a}^{b-1}\omega(x-1)P_n^2(x)-d_n^2.
\]
Finally, taking into account the non trivial expression \cite{alvareznodarse_03,nikiforov_88,nikiforov_91}
\[
\frac{\omega(x-1)}{\omega(x)}=\frac{\sigma(x)}{\tau(x-1)+\sigma(x-1)},
\]
we obtain the result
\[
\sum_{x=a}^{b-1} \omega(x) \Delta P_n^2(x)=\left.\omega(x-1)P_n^2(x)\right|_a^b+\left\langle \frac{\sigma(x)}{\tau(x-1)+\sigma(x-1)}\right\rangle -d_n^2,
\]
which together with Eq. (\ref{eq:fisherdefinition}) yields the wanted expression (\ref{eq:fishersigmatau}), and the theorem is proved.
\end{proof}

\section{The relative Fisher information of the classical discrete polynomials}
\label{sec3}

In this Section we use the ladder-relation-based method to calculate the relative Fisher information of the four classical families of discrete orthogonal polynomials: Charlier, Meixner, Kravchuk and Hahn polynomials. Of course one can alternatively choose the difference-equation-based method, obtaining the same results but the use of the former method is simpler in general; this is especially true for the Hahn case. For convenience, we gather in Table \ref{tab:discrete} the main data for the classical discrete orthogonal polynomials; namely, the orthogonality interval $(a,b)$, the weight function $\omega(x)$ and the coefficients of the second order difference equation. For ease of reading, the ladder relation is indicated below for each concrete family.

\subsection{Charlier polynomials}

The Charlier polynomials $\{C_n^\mu(x)\}$ are orthogonal with respect to a discrete measure whose distribution function has jumps $\frac{\mu^xe^{-\mu}}{x!}$ at $x=0,1,\ldots$, where $\mu>0$, so that its associated weight function is the Poisson distribution
\[
\omega(x)=\frac{\mu^xe^{-\mu}}{x!}
\]
The monic Charlier polynomials fulfil the ladder relation
\[
\Delta C_n^\mu(x)=nC_{n-1}^\mu(x),
\]
so that one has 
\[
I_\omega[C_n^\mu]=\frac{1}{d_n^2}\sum_{x=0}^\infty \omega(x)[\Delta C_n^\mu(x)]^2=
\frac{1}{d_n^2}\sum_{x=0}^\infty \omega(x)n^2[C_{n-1}^\mu(x)]^2=n^2\frac{d_{n-1}^2}{d_n^2}.
\]
Then, taking into account the normalization constant of these polynomials given in Table \ref{tab:discrete} we have finally the following simple values 
\begin{equation}
I_\omega[C_n^\mu]=\frac{n}{\mu},
\label{eq:fishercharlier}
\end{equation}
for the relative Fisher information of the Charlier polynomials.

\subsection{Meixner polynomials}

The Meixner polynomials $\{M_n^{\gamma,\mu}(x)\}$ are orthogonal with respect to a discrete measure whose distribution function has jumps $\frac{\mu^x\Gamma(\gamma+x)}{x!\Gamma(\gamma)}$ at $x = 0,1,\ldots$, where $\gamma > 0$. For integrability and positivity of the measure we need that $\mu\in (0,1)$, so that the corresponding weight function is the negative binomial or Pascal distribution
\[
\omega(x)=\frac{\mu^x\Gamma(\gamma+x)}{x!\Gamma(\gamma)}.
\]
The ladder relation of the monic Meixner polynomials is
\[
\Delta M_n^{\gamma,\mu}(x)=n M_{n-1}^{\gamma+1,\mu}(x).
\]
Moreover, \'Alvarez-Nodarse et al (see \cite{alvareznodarse:jcam97}, subsection 5.1.2) found that these polynomials satisfy the general connection formula
\begin{equation}
M_m^{\eta,\zeta}(x)=\sum_{j=0}^m c_{m,j} M_j^{\alpha,\beta}(x),
\label{eq:expansionmeixner}
\end{equation}
with the values 
\begin{eqnarray}
c_{m,j}=\left(
\begin{array}{c}
m\\
j
\end{array}
\right)
\left(\frac{\zeta}{\zeta-1}\right)^{m-j}(\eta+j)_{m-j}\,_2F_1\left(\left.
\begin{array}{c}
j-m,j+\alpha\\
j+\eta
\end{array}
\right|\frac{\beta(1-\zeta)}{\zeta(1-\beta)}\right)
\label{eq:coefmeixner}
\end{eqnarray}
for the expansion coefficients. To obtain the necessary expansion (\ref{eq:linearization}) of $\Delta M_n^{\gamma,\mu}(x)$, we need the particular case of formulas (\ref{eq:expansionmeixner})-(\ref{eq:coefmeixner}) with the parameters $\zeta = \beta = \mu$, $\eta = \gamma + 1$, $\alpha = \gamma$ and $m = n-1$. Then, one has the expansion coefficients
\[
c_{n-1,j}=(j+1)_{n-j-1}\left(\frac{\mu}{\mu-1}\right)^{n-j-1},
\]
so that the wanted expansion of $\Delta M_n^{\gamma,\mu}(x)$ is
\begin{equation}
\Delta M_{n}^{\gamma,\mu}(x)=n\sum_{j=0}^{n-1}(j+1)_{n-j-1}\left(\frac{\mu}{\mu-1}\right)^{n-j-1} M_j^{\gamma,\mu}(x).
\label{eq:expansionmeixner2}
\end{equation}

Then, the relative Fisher information of the Meixner polynomials is 
\begin{eqnarray*}
I_\omega[M_n^{\gamma,\mu}]&=&\frac{1}{d_n^2}\sum_{x=0}^\infty \omega(x)[\Delta M_n^{\gamma,\mu}(x)]^2\\
&=&
\frac{1}{d_n^2}\sum_{x=0}^\infty \omega(x)\left[n\sum_{j=0}^{n-1}(j+1)_{n-j-1}\left(\frac{\mu}{\mu-1}\right)^{n-j-1} M_j^{\gamma,\mu}(x)\right]^2\\
&=&
\frac{1}{d_n^2} n^2 \sum_{j=0}^{n-1} \left[(j+1)_{n-j-1}\left(\frac{\mu}{\mu-1}\right)^{n-j-1}\right]^2 d_j^2,
\end{eqnarray*}
where we have taken into account the expressions (\ref{eq:fisherdefinition}), (\ref{eq:expansionmeixner2}) and (\ref{eq:orthogonality}), respectively. Finally, with the normalization constant $d_n^2$ given in Table \ref{tab:discrete} we have the value 
\begin{eqnarray}
I_\omega[M_n^{\gamma,\mu}]=n\frac{(1-\mu)^2}{\mu(n+\gamma-1)}\,_2F_1\left(\left.
\begin{array}{c}
1-n,1\\
2-n-\gamma
\end{array}
\right|\mu\right)
\label{eq:fishermeixner}
\end{eqnarray}
for the relative Fisher information of the Meixner polynomials.

\subsection{Kravchuk polynomials}

The Kravchuk polynomials  $\{K_n^p(x,N)\}$ are orthogonal with respect to the discrete weight function $\omega(x)$, $x = 0, 1,\ldots,N$, given by the binomial distribution
\[
\omega(x)=\binom{N}{x} p^x(1-p)^{N-x},
\]
with  $0 < p < 1$  and $n\leq N-1$. The ladder relation for the monic Kravchuk polynomials is
\begin{equation}
\Delta K_n^{p}(x,N)=n K_{n-1}^{p}(x,N-1).
\label{eq:ladderkravchuk}
\end{equation}

These polynomials fulfil the following general connection formula (see \cite{alvareznodarse:jcam97}, subsection 5.1.3)
\[
K_m^{p}(x,M)=\sum_{j=0}^m c_{m,j} K_j^{q}(x,N),
\]
with the following expansion coefficients
\begin{eqnarray*}
c_{m,j}=\left(
\begin{array}{c}
m\\
j
\end{array}
\right)
(M-m+1)_{m-j}(-p)^{m-j}
\,_2F_1\left(\left.
\begin{array}{c}
j-m,j-N\\
j-M
\end{array}
\right|\frac{q}{p}\right).
\end{eqnarray*}

Keeping in mind Eq. (\ref{eq:ladderkravchuk}), we are interested in the particular case ($q = p$, $M = N-1$, $m = n-1$) of this connection formula so that we have
\begin{eqnarray*}
c_{n-1,j}=\left(
\begin{array}{c}
n-1\\
j
\end{array}
\right)
p^{n-j-1}(n-j-1)!=(j+1)_{n-j-1} p^{n-j-1}.
\end{eqnarray*}
Then, we have the wanted expansion of $\Delta K_n^p(x,N)$ as
\begin{equation}
\Delta K_{n}^{p}(x,N)=n\sum_{j=0}^{n-1}(j+1)_{n-j-1}p^{n-j-1} K_j^{p}(x,N),
\label{eq:expansionkravchuk2}
\end{equation}
which corresponds to the expansion (5) for Kravchuk polynomials. Then, the relative Fisher information of these objects is 
\begin{eqnarray*}
I_\omega[K_n^{p}(N)]&=&\frac{1}{d_n^2}\sum_{x=0}^{N} \omega(x)[\Delta K_n^{p}(x,N)]^2\\
&=&
\frac{1}{d_n^2}\sum_{x=0}^{N} \omega(x)\left[n\sum_{j=0}^{n-1}(j+1)_{n-j-1}p^{n-j-1} K_j^{p}(x,N)\right]^2\\
&=&
\frac{1}{d_n^2} n^2 \sum_{j=0}^{n-1} \left[(j+1)_{n-j-1}p^{n-j-1}\right]^2 d_j^2,
\end{eqnarray*}
where we have taken into account the expressions (\ref{eq:fisherdefinition}), (\ref{eq:expansionkravchuk2}) and (\ref{eq:orthogonality}), respectively. Finally, with the normalization constant $d_n^2$ given in Table \ref{tab:discrete} we have the value
\begin{eqnarray}
I_\omega[K_n^{p}(N)]=\frac{n}{N-n+1}\frac{1}{p(1-p)}
\,_2F_1\left(\left.
\begin{array}{c}
1-n,1\\
N-n+2
\end{array}
\right|\frac{p}{p-1}\right)
\label{eq:fisherkravchuk}
\end{eqnarray}
for the relative Fisher information of the Kravchuk polynomials.

\subsection{Hahn polynomials}

The Hahn polynomials $\{h_n^{\alpha,\beta}(x,N)\}$ are orthogonal on $[0, N-1]$ with respect to the hypergeometric distribution
\[
\omega(x)=\frac{\Gamma(N+\alpha-x)\Gamma(\beta+x+1)}{(N-x-1)!x!}
\]
with $\alpha > -1$ and $\beta > -1$. The ladder relation fulfilled by the monic Hahn polynomials is 
\begin{equation}
\Delta h_n^{\alpha,\beta}(x,N)=n h_{n-1}^{\alpha+1,\beta+1}(x,N-1).
\label{eq:ladderhahn}
\end{equation}
Moreover, they satisfy the following general connection formula (see \cite{alvareznodarse:jcam97}, subsection 5.1.10)
\[
h_m^{\gamma,\mu}(x,M)=\sum_{j=0}^m c_{m,j} h_j^{\alpha,\beta}(x,N),
\]
with the expansion coefficients
\begin{eqnarray*}
c_{m,j}&=&\left(
\begin{array}{c}
m\\
j
\end{array}
\right)
\frac{(1+j-M)_{m-j}(1+j+\mu)_{m-j}}{(1+j+m+\gamma+\mu)_{m-j}}\\
&&\times\,_4F_3\left(\left.
\begin{array}{c}
m-j,1+j-N,j+\beta+1,1+m+j+\gamma+\mu\\
1+j-M,j+\mu+1,2j+\alpha+\beta+2
\end{array}
\right|1\right).
\end{eqnarray*}

In our case, keeping in mind (\ref{eq:ladderhahn}), we are interested in the particular connection problem with $M = N-1$, $m = n-1$, $\gamma = \alpha +1$ and $\mu = \beta +1$. The corresponding expansion coefficients are 
\begin{multline}
c_{n-1,j}=\left(
\begin{array}{c}
n-1\\
j
\end{array}
\right)
\frac{(2+j-N)_{n-1-j}(2+j+\beta)_{n-1-j}}{(2+j+n+\alpha+\beta)_{n-1-j}}\\
\times\,_4F_3\left(\left.
\begin{array}{c}
j-n+1,1+j-N,j+\beta+1,2+n+j+\alpha+\beta\\
2+j-N,j+\beta+2,2j+\alpha+\beta+2
\end{array}
\right|1\right),
\label{eq:coefhahn}
\end{multline}
so that the wanted expansion of $\Delta h_n^{\alpha,\beta}(x)$ is
\begin{multline}
\label{eq:expansionhahn2}
\Delta h_n^{\alpha,\beta}(x,N)=n\sum_{j=0}^{n-1}  \left(
\begin{array}{c}
n-1\\
j
\end{array}
\right)
\frac{(2+j-N)_{n-1-j}(2+j+\beta)_{n-1-j}}{(2+j+n+\alpha+\beta)_{n-1-j}}\\
\times\,_4F_3\left(\left.
\begin{array}{c}
j-n+1,1+j-N,j+\beta+1,2+n+j+\alpha+\beta\\
2+j-N,j+\beta+2,2j+\alpha+\beta+2
\end{array}
\right|1\right)\\
\times h_j^{\alpha,\beta}(x,N),
\end{multline}
which corresponds to the expansion (\ref{eq:linearization}) for Hahn polynomials. Then, the relative Fisher information of these objects is 
\begin{multline*}
I_\omega[h_n^{\alpha,\beta}(N)]=\frac{1}{d_n^2}\sum_{x=0}^{N-1}\omega(x)[\Delta h_n^{\alpha,\beta}(x,N)]^2\\
=\frac{1}{d_n^2}\sum_{x=0}^N\omega(x)\left[n\sum_{j=0}^{n-1}c_{n-1,j} h_j^{\alpha,\beta}(x,N)\right]^2
=\frac{1}{d_n^2}n^2\sum_{j=0}^{n-1} c_{n-1,j}^2 d_j^2,
\end{multline*}
where we have taken into account the expressions (\ref{eq:fisherdefinition}), (\ref{eq:expansionhahn2}) and (\ref{eq:orthogonality}), respectively. Finally, with the normalization constant $d_n^2$ given in Table \ref{tab:discrete} and constant $c_{n-1,j}$ given in Eq. (\ref{eq:coefhahn}) we have the value
\begin{equation}
I_\omega[h_n^{\alpha,\beta}(N)] = A_1 A_2(B_1 B_2 B_3 +C_1C_2C_3+D_1D_2D_3)
\label{eq:fisherhahn}
\end{equation}
for the relative Fisher information of the Hahn polynomials, where
\[
A_1=\frac{n^2(\alpha+\beta+2n+1)(N-n-1)!\Gamma(\alpha+\beta+n+1)}{n!\Gamma(\alpha+n+1)\Gamma(\beta+n+1)\Gamma(\alpha+\beta+N+n+1)},
\]
\[
A_2=\frac{((\alpha+\beta+n+1)_n)^2\Gamma(\alpha+1)\Gamma(\beta+1)\Gamma(\alpha+\beta+N+1)}{(\alpha+\beta+1)(N-1)!\Gamma(\alpha+\beta+1)},
\]
\[
B_1=\left(\frac{(n-1)!(\beta+1)(\alpha+\beta+N+1)(-\alpha-\beta-n-N)_{n-1}(\beta+2)_{n-1}}{(\alpha+\beta+n+2)_{n-1}(-\alpha-\beta-n-1)_{n-1}(\alpha+\beta+2)(N+\beta)}\right)^2,
\]
\[
B_2=\frac{(-1)^{n-1}(\alpha+1)_{n-1}\left(\frac{\alpha+\beta+3}{2}\right)_{n-1}(\alpha+\beta+1)_{n-1}(1-N)_{n-1}}{(n-1)!\left(\frac{\alpha+\beta+1}{2}\right)_{n-1}(\beta+1)_{n-1}(\alpha+\beta+N+1)_{n-1}},
\]
\[
B_3=\,_5F_4\left(\left.
\begin{array}{c}
1-n,1,1-n-\beta,1-n-\alpha-\beta-N,2-n-\frac{\alpha+\beta+1}{2}\\
1-n-\alpha,2-n-\frac{\alpha+\beta+3}{2},1-n-\alpha-\beta,1-n+N
\end{array}
\right|-1\right),
\]
\[
C_1=\frac{2(-1)^n((n-1)!)^2(\beta+1)(\alpha+\beta+N+1)(-\alpha-\beta-n-N)_{n-1}}{(\alpha+\beta+n+2)_{n-1}^2(-\alpha-\beta-n-1)_{n-1}^2(\alpha+\beta+2)^2},
\]
\[
C_2=\frac{(\beta+2)_{n-1}(1-N)(\alpha+1)(-\alpha-n)_{n-1}(2-N)_{n-1}(\alpha+\beta+2n+1)}{n!(-N-\beta)^2\Gamma(\alpha+\beta+2)},
\]
\[
C_3=\Gamma(\alpha+\beta+n+1)\,_3F_2\left(\left.
\begin{array}{c}
1,\frac{\alpha+\beta+3}{2}+n,\alpha+\beta+n+1\\
n+1,\frac{\alpha+\beta+1}{2}+n
\end{array}
\right|-1\right),
\]
\[
D_1=\left(\frac{(n-1)!(N-1)(\alpha+1)(-\alpha-n)_{n-1}(-N+2)_{n-1}}{(\alpha+\beta+n+2)_{n-1}(-\alpha-\beta-n-1)_{n-1}(\alpha+\beta+2)(N+\beta)}\right)^2,
\]
\[
D_2=\frac{(-1)^{n-1}\left(\frac{\alpha+\beta+3}{2}\right)_{n-1}(\beta+1)_{n-1}(\alpha+\beta+N+1)_{n-1}(\alpha+\beta+1)_{n-1}}{(n-1)!(1-N)_{n-1}(\alpha+1)_{n-1}\left(\frac{\alpha+\beta+1}{2}\right)_{n-1}},
\]
\[
D_3=\,_5F_4\left(\left.
\begin{array}{c}
1-n,1,1-n+N,1-n-\alpha,2-n-\frac{\alpha+\beta+1}{2}\\
2-n-\frac{\alpha+\beta+3}{2},1-n-\beta,1-n-\alpha-\beta-N,1-n-\alpha-\beta
\end{array}
\right|-1\right).
\]
These expressions have been obtained after a long, but not very difficult, simplifying and rewriting process, where key equations (5, page 438), (15, page 535) and (1, page 552) from \cite{prudnikov_86} have been employed.

\section{Limiting cases and numerical studies}
\label{sec4}

Here we will analyze the relative Fisher information of the classical discrete polynomials in terms of the degree and parameters which characterize them. First of all, we observe from Eq. (\ref{eq:fishercharlier}) that the relative Fisher information of Charlier polynomials $I_\omega[C_n^\mu]$ depends linearly on $n$ and inversely on $\mu$. The expressions (\ref{eq:fishermeixner}), (\ref{eq:fisherkravchuk}) and (\ref{eq:fisherhahn}) for the relative Fisher information of Meixner, Kravchuk and Hahn polynomials are much more complicated and, therefore, we discuss them in various limiting cases.

\subsection{Meixner polynomials}

According to Eq. (\ref{eq:fishermeixner}) and the properties of the involved hypergeometric function, the dependence of the relative Fisher information of the Meixner polynomials $I_\omega[M_n^{\gamma,\mu}]$ on the degree $n$ and the parameters $\gamma$ and $\mu$ have the following characteristics. First,
it grows with the degree $n$ for $\gamma>1$ and it decreases when $n$ increases for $\gamma<1$,
with the following asymptotic behaviour, obtained by considering its Taylor expansion when $n\to\infty$
\[
I_\omega[M_n^{\gamma,\mu}]\sim \frac{1-\mu}{\mu}+\frac{1-\gamma}{n},\; (n\to \infty).
\]
Moreover, the larger the value of $\gamma$ is, the slower is the growth of the Fisher information. This behaviour can be observed in Figure \ref{fig:meix1}, where the relative Fisher information for the polynomials $M_n^{3/2,1/4}$, $M_n^{4,1/4}$ and $M_n^{3/2,1/7}$, are represented as a function of the degree $n$.
Indeed, this quantity tends towards 3 and 6 for the cases with $\mu=\frac14$ and $\mu=\frac17$, respectively; and, moreover, its increasing rate is slower for $\gamma=4$ than for $\gamma=\frac32$.

As a function of the parameter $\mu$, the Fisher information has the following asymptotic behaviours,
\begin{eqnarray*}
I_\omega[M_n^{\gamma,\mu}]\sim \frac{n}{n+\gamma-1}\,_2F_1\left(\left.
\begin{array}{c}
1-n,1\\
2-n-\gamma
\end{array}
\right|1\right)
(1-\mu)^2, \; (\mu\to 1),
\end{eqnarray*}
\[
I_\omega[M_n^{\gamma,\mu}]\sim\frac{n}{n+\gamma-1} \frac{1}{\mu}, \; (\mu\to 0).
\]
Thus, when $\mu$ tends to 1, the larger the value of $n$ or the smaller the value of $\gamma$ are, the faster the relative Fisher information approaches zero. And when $\mu$ tends to 0, the larger the value of $n$ or the smaller the value of $\gamma$ are, the faster the relative Fisher information approaches infinity. This is observed in Figure \ref{fig:meix2} where the Fisher information for the polynomials $M_2^{3/2,\mu}$, $M_2^{4,\mu}$ and $M_5^{3/2,\mu}$, as a function of the parameter $\mu$.
These results suggest that
\[
I_\omega[M_{n_1}^{\gamma,\mu}] < I_\omega[M_{n_2}^{\gamma,\mu}] \iff n_1<n_2,
\]
and
\[
I_\omega[M_{n}^{\gamma_1,\mu}] < I_\omega[M_{n}^{\gamma_2,\mu}] \iff \gamma_1>\gamma_2.
\]

As a function of the parameter $\gamma$, the relative Fisher information have the following asymptotic behaviours,
\[
I_\omega[M_n^{\gamma,\mu}]\sim \frac{n(1-\mu)^2}{\mu}\frac{1}{\gamma},\; (\gamma\to +\infty),
\]
\[
I_\omega[M_n^{\gamma,\mu}]\sim \frac{n}{\gamma}(1-\mu)^2\mu^{n-2},\; (\gamma\to 0).
\]
Thus, when $\gamma$ tends to infinity, the larger the value of $n$ or the smaller the value of $\mu$ are, the slower the Fisher information approaches zero. And when $\gamma$ tends to zero, the rate of growth of the relative Fisher information depends on the value of $n(1-\mu)^2 \mu^{n-2}$. This is observed in Figure \ref{fig:meix3} where the Fisher information of the polynomials $M_2^{\gamma,1/4}$, $M_2^{\gamma,3/4}$ and $M_5^{\gamma,1/4}$ are plotted as a function of the parameter $\gamma$.

\subsection{Kravchuk polynomials}

A similar analysis for the expression (\ref{eq:fisherkravchuk}) of the relative Fisher information of the Kravchuk polynomials $I_\omega(K_n^p)$ first shows that this quantity grows when the degree $n$ is increasing, being $(p,N)$ fixed, so that for the maximum degree (namely, $n = N - 1$) it has the value

\[
I_\omega[K_{N-1}^p(N)]=\frac{(1-p)^{1-N}+(1-N)p-1}{Np^3}.
\]
Moreover, when $N$ tends to infinity, this maximum value increases as
\[
I_\omega[K_{N-1}^p(N)]\sim \frac{1}{N(1-p)^{N-1}p^3}.
\]

Furthermore, the dependence of the Fisher quantity on the degree $n$ is numerically shown for the polynomials $K_n^{1/7}(15)$, $K_n^{1/7}(20)$ and $K_n^{1/3}(15)$ in Figure \ref{fig:kra1}. Therein we observe not only its growth with $n$, but also that its rate of increasing is slowing down when the parameter $N$ gets bigger, being $p$ fixed, or when the parameter $p$ gets smaller, being $N$ fixed. This behaviour can also be obtained from expression (\ref{eq:fisherkravchuk}). Thus, these results suggest that
\[
I_\omega[K_n^p(N_1)]<I_\omega[K_n^p(N_2)] \iff N_1>N_2,
\]
and
\[
I_\omega[K_n^{p_1}(N)]<I_\omega[K_n^{p_2}(N)] \iff p_1<p_2.
\]

The dependence of $I_\omega[K_n^p(N)]$ on the parameter $p$ when $n$ and $N$ are fixed is further studied for the polynomials $K_2^p(15)$, $K_4^p(15)$ and $K_2^p(20)$ in Figure \ref{fig:kra2}. We remark that it has a concave shape, with a minimum around $p=1/2$ but with a different asymptotic behaviour at its two extremes; namely,
the asymptotic behaviour when $p$ tends to 0 is
\[
I_\omega[K_n^p(N)]\sim \frac{n}{N-n+1}\frac{1}{p},
\]
and when $p$ tends to 1 is
\[
I_\omega[K_n^p(N)]\sim \frac{n!}{(N-n+1)_{n}}\frac{1}{(1-p)^n}.
\]

As a function of the parameter $N$, the relative Fisher information decreases as the value of $N$ increases. This is numerically shown in Figure \ref{fig:kra3} for the Fisher information of the polynomials $K_2^{1/7}(N)$, $K_4^{1/7}(N)$ and $K_2^{1/3}(N)$, is represented as a function of the value of $N$.

\subsection{Hahn polynomials}

The complexity of the expression (\ref{eq:fisherhahn}) hardly allows for an analysis analogous to that of the Meixner and Kravchuk polynomials. However, the numerical analysis of the relative Fisher information of the Hahn polynomials $I_\omega[h_n^{\alpha,\beta}(N)]$ first shows that this quantity grows with the degree $n$, being $\alpha$, $\beta$ and $N$ fixed. Moreover, like in the Kravchuk case, the maximum is obtained for $n=N-1$, whose value depends on $\alpha$, $\beta$ and $N$. This behaviour can be observed in Figure \ref{fig:hahn1}, where the relative Fisher information is represented as a function of the degree $n$ for the polynomials $h_n^{0,0}(20)$, $h_n^{0,0}(30)$ and $h_n^{3,-1/2}(20)$.

The dependence of $I_\omega[h_n^{\alpha,\beta}(N)]$ on the parameter $N$ is considered in Figure \ref{fig:hahn2}, where the relative Fisher information is represented as a function of $N$ for the polynomials $h_2^{0,0}(N)$, $h_{10}^{0,0}(N)$ and $h_2^{3,-1/2}(N)$. Thus, we observe that $I_\omega[h_n^{\alpha,\beta}(N)]$ decreases when $N$ increases, being $n$, $\alpha$ and $\beta$ fixed.

The relative Fisher information is represented in Figure \ref{fig:hahn3} as a function of the parameter $\alpha$ for the polynomials $h_2^{\alpha,0}(20)$, $h_2^{\alpha,3}(20)$, $h_{10}^{\alpha,0}(20)$ and $h_2^{\alpha,0}(30)$. Herein, we observe the expected divergent behaviour when $\alpha\to -1$, and a linear asymptotic behaviour for $\alpha\to +\infty$. The dependence on the parameter $\beta$ is similar to that on $\alpha$, being divergent when $\beta\to -1$ and asymptotically linear for $\beta\to +\infty$. Moreover, these behaviours are much more emphasized for $\beta$ than for $\alpha$. This can be observed in Figure \ref{fig:hahn4}, where the relative Fisher information is represented as a function of $\beta$ for the polynomials $h_2^{0,\beta}(20)$, $h_2^{3,\beta}(20)$, $h_{10}^{0,\beta}(20)$ and $h_2^{0,\beta}(30)$.

\section{Conclusions}

This work is a contribution to the information theory of both the discrete distributions and the special functions of applied mathematics and mathematical physics. Here, we have studied both analytical and numerically the relative Fisher information of the classical orthogonal polynomials in a discrete variable (i.e., the Charlier, Meixner, Kravchuk and Hahn polynomials) with respect to their weight functions. This quantity measures the separation between the Rakhmanov density associated to these polynomials and their respective weight functions. From a technical point of view, it can also be interpreted as a measure of the pointwise concentration of  the probability cloud associated to the polynomial under consideration. As well, it can be viewed as a quantifier of the oscillatory character of the polynomial itself and its corresponding Rakhmanov density.

Two different analytical methods based on the ladder relation and the second order difference relation, respectively, satisfied by these polynomials have been proposed. Then, the ladder-relation-based method is used to find close expressions for the Fisher quantity in all four families. Finally, the resulting general expressions are numerically investigated in terms of the degree and the involved parameters of the polynomials.

\section*{Acknowledgements}

This work has been partially supported by Junta de Andalucia under grants FQM-4643 and FQM-2445, as well the Ministerio de Ciencia e Innovaci\'on under grant FIS2008-02380/FIS. We belong to the Andalusian research group FQM-207. We are grateful to Roberto Costas valuable discussions.

\newpage

\begin{table}
\begin{center}
\begin{tabular}{|c|c|c|c|}
\hline
\Ti   &  {\bf Charlier}  & {\bf Meixner}  & {\bf Kravchuk} \\
$y_n(x)$ &  $C_n^{ \mu}(x)$ & $M_n^{\gamma,\mu}(x)$ &
$K_n^{p }(x)$  \\
\hline \hline
\Ti\Bi $\sigma(x) $ & $x$  & $x$  & $x$  \\
\Ti\Bi $\tau(x)$ & $\mu-x$ &  $(\mu-1)x+\mu\gamma$  &
 $ \displaystyle \frac{Np-x}{1-p} $  \\
\Ti\B   $\lambda_n$ & $n$ & $(1-\mu)n $ & $ \displaystyle \frac{n}{1-p}$ \\
\hline
\Ti\Bi$\, [a,b] \, $ & $[0,\infty)$ & $[0,\infty)$ & $[0,N]$ \\
\Ti\Bi $\omega(x)$& $  \displaystyle \frac{e^{-\mu} \mu^x}{\Gamma(x+1)}$ &
  $  \displaystyle \frac{\mu^x  \Gamma(\gamma+x)}{\Gamma(\gamma)\Gamma(x+1)}$  &
\Ti\Bi $\displaystyle \frac{N! p^x (1-p)^{N-x}}{\Gamma(N+1-x)\Gamma(x+1)}$  \\
\Ti\Bi  &  ($\mu >0$) & ($\gamma>0,0<\mu<1$)  & ($0<p<1$, $n\leq N-1$)  \\
\hline
\T\B $d_n^2$ & $n! \mu^n$ & $ \displaystyle \frac{n!(\gamma)_n \mu^n}{(1-\mu)^{\gamma+2n}}$  &
$ \displaystyle \frac{n!N! p^n (1-p)^n}{(N-n)!}$  \\
\hline\hline
\Ti  & \multicolumn{3}{c|}{\textbf{Hahn}}\\
$y_n(x)$ &  \multicolumn{3}{c|}{$h_n^{\alpha,\beta}(x;N)$}  \\
\hline \hline
\Ti  $\sigma(x) $ & \multicolumn{3}{c|}{$x(N+\alpha -x)$}  \\
\Ti $\tau(x)$ & \multicolumn{3}{c|}{$(\beta+1)(N-1)-(\alpha+\beta+2)x$}\\
\Ti   $\lambda_n$ &  \multicolumn{3}{c|}{$ n(n+\alpha+\beta+1)$}  \\
\hline
\Ti\Bi $\, [a,b] \, $ & \multicolumn{3}{c|}{$[0,N-1]$}  \\
\Ti\Bi  $\omega(x)$& \multicolumn{3}{c|}{$ \displaystyle \frac{ \Gamma(N+\alpha-x)\Gamma(\beta+x+1)}{\Gamma(N-x)\Gamma(x+1)}$} \\
\Ti\Bi  &  \multicolumn{3}{c|}{(${\alpha,\beta \geq -1 \ , \ n\leq N-1}$)}   \\
\hline
\T\B $d_n^2$ & \multicolumn{3}{c|}{$ \displaystyle \frac{n! \Gamma(n+\alpha+1) \Gamma(n+\beta+1) \Gamma(N+n+\alpha+\beta+1)}{(2n+\alpha+\beta+1)(N-n-1)! \Gamma(n+\alpha+\beta+1) (n+\alpha+\beta+1)_n^2}$} \\
\hline
\end{tabular}
\caption{Coefficients $\sigma(x)$, $\tau(x)$ and $\lambda_n$ of the difference equation (\ref{eq:differenceequation}), interval of definition $[a,b]$, orthogonality weight $\omega(x)$, and norm $d_n^2$ for the classical families of monic discrete orthogonal polynomials.}
\label{tab:discrete}
\end{center}
\end{table}

\newpage

\begin{figure}
\begin{center}
\includegraphics[height=11cm,angle=270]{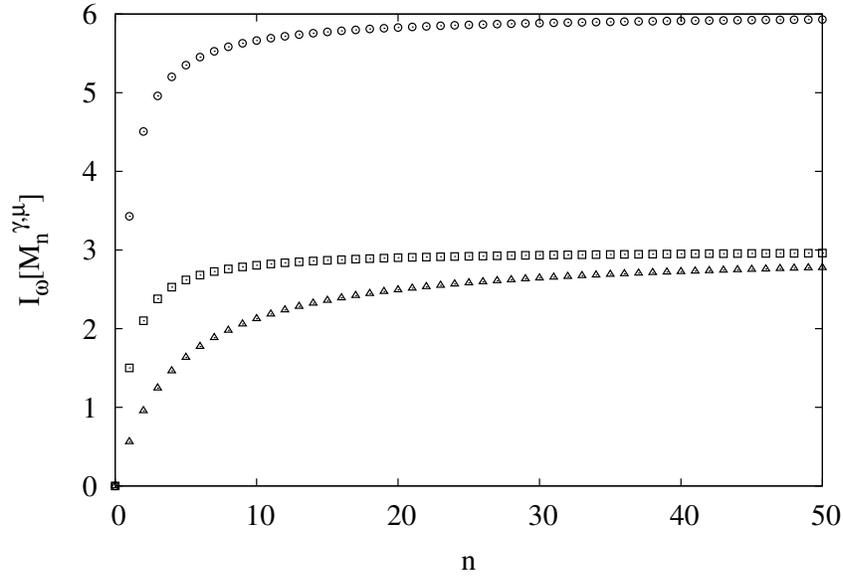}
\caption{Relative Fisher informations $I_\omega[M_n^{3/2,1/4}]$ ($\square$), $I_\omega[M_n^{4,1/4}]$ ($\vartriangle$) and $I_\omega[M_n^{3/2,1/7}]$ ($\bigcirc$) as a function of the degree $n$.}
\label{fig:meix1}\end{center}
\end{figure}

\begin{figure}
\begin{center}
\includegraphics[height=11cm,angle=270]{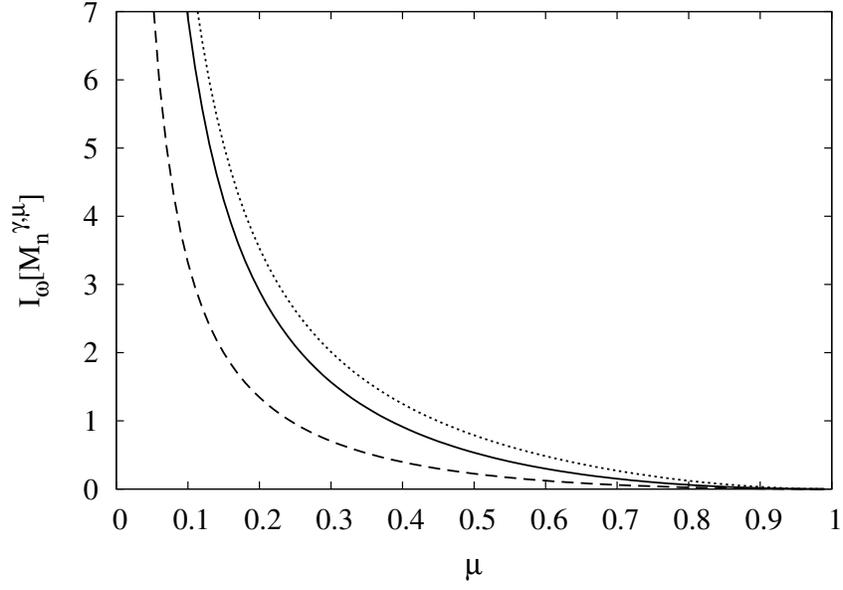}
\caption{Relative Fisher informations $I_\omega[M_2^{3/2,\mu}]$ (solid line), $I_\omega[M_2^{4,\mu}]$ (dashed line) and $I_\omega[M_5^{3/2,\mu}]$ (dotted line) as a function of the parameter $\mu$.}
\label{fig:meix2}\end{center}
\end{figure}

\begin{figure}
\begin{center}
\includegraphics[height=11cm,angle=270]{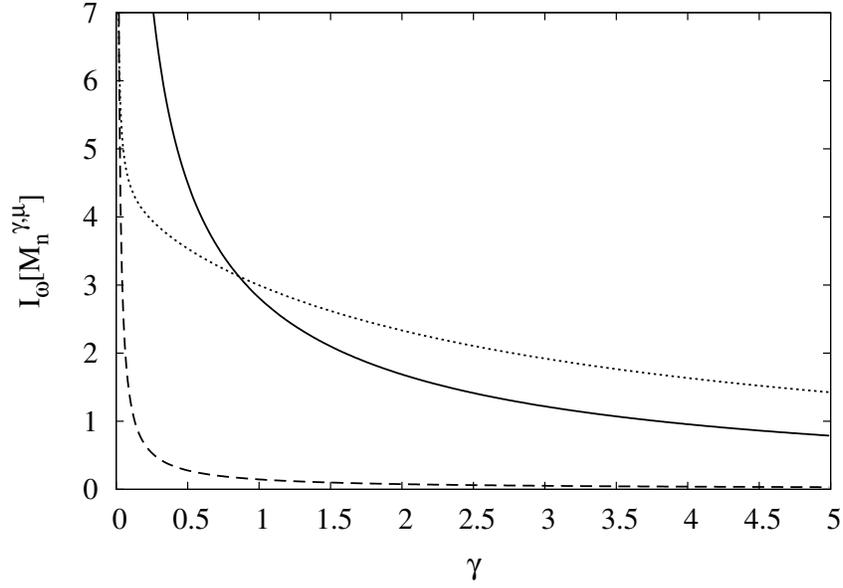}
\caption{Relative Fisher informations $I_\omega[M_2^{\gamma,1/4}]$ (solid line), $I_\omega[M_2^{\gamma,3/4}]$ (dashed line) and $I_\omega[M_5^{\gamma,1/4}]$ (dotted line) as a function of the parameter $\gamma$.}
\label{fig:meix3}\end{center}
\end{figure}

\begin{figure}
\begin{center}
\includegraphics[height=11cm,angle=270]{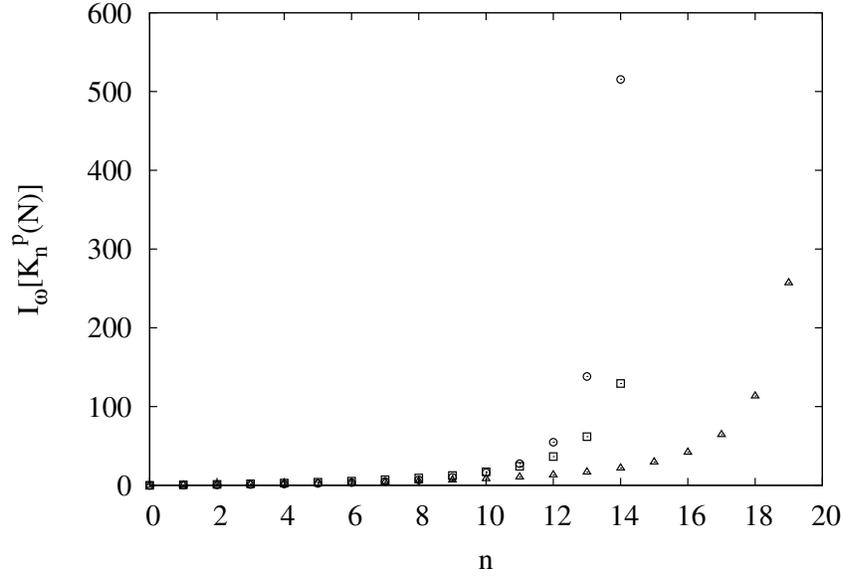}
\caption{Relative Fisher informations $I_\omega[K_n^{1/7}(15)]$ ($\square$), $I_\omega[K_n^{1/7}(20)]$ ($\vartriangle$) and $I_\omega[K_n^{1/3}(15)]$ ($\bigcirc$) as a function of the degree $n$.}
\label{fig:kra1}
\end{center}
\end{figure}

\begin{figure}
\begin{center}
\includegraphics[height=11cm,angle=270]{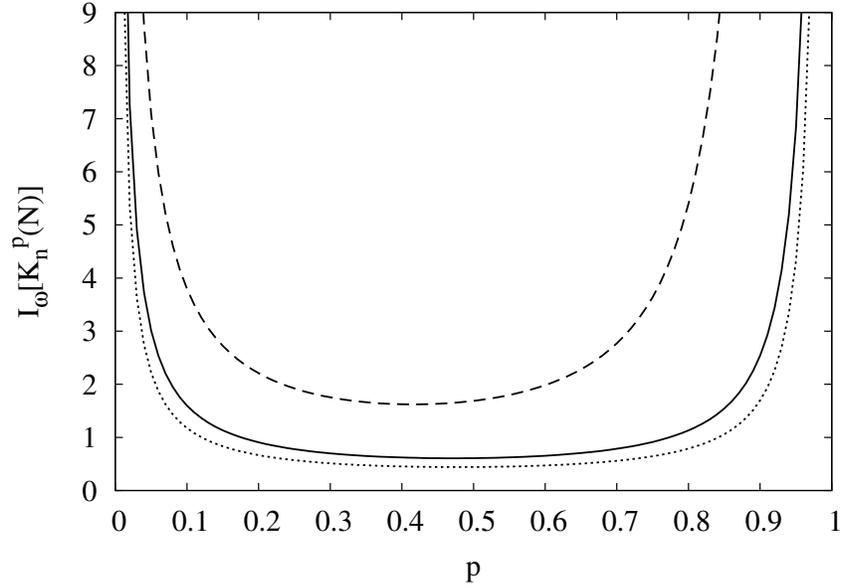}
\caption{Relative Fisher informations $I_\omega[K_2^{p}(15)]$ (solid line), $I_\omega[K_4^{p}(15)]$ (dashed line) and $I_\omega[K_2^{p}(20)]$ (dotted line) as a function of the parameter $p$.}
\label{fig:kra2}
\end{center}
\end{figure}

\begin{figure}
\begin{center}
\includegraphics[height=11cm,angle=270]{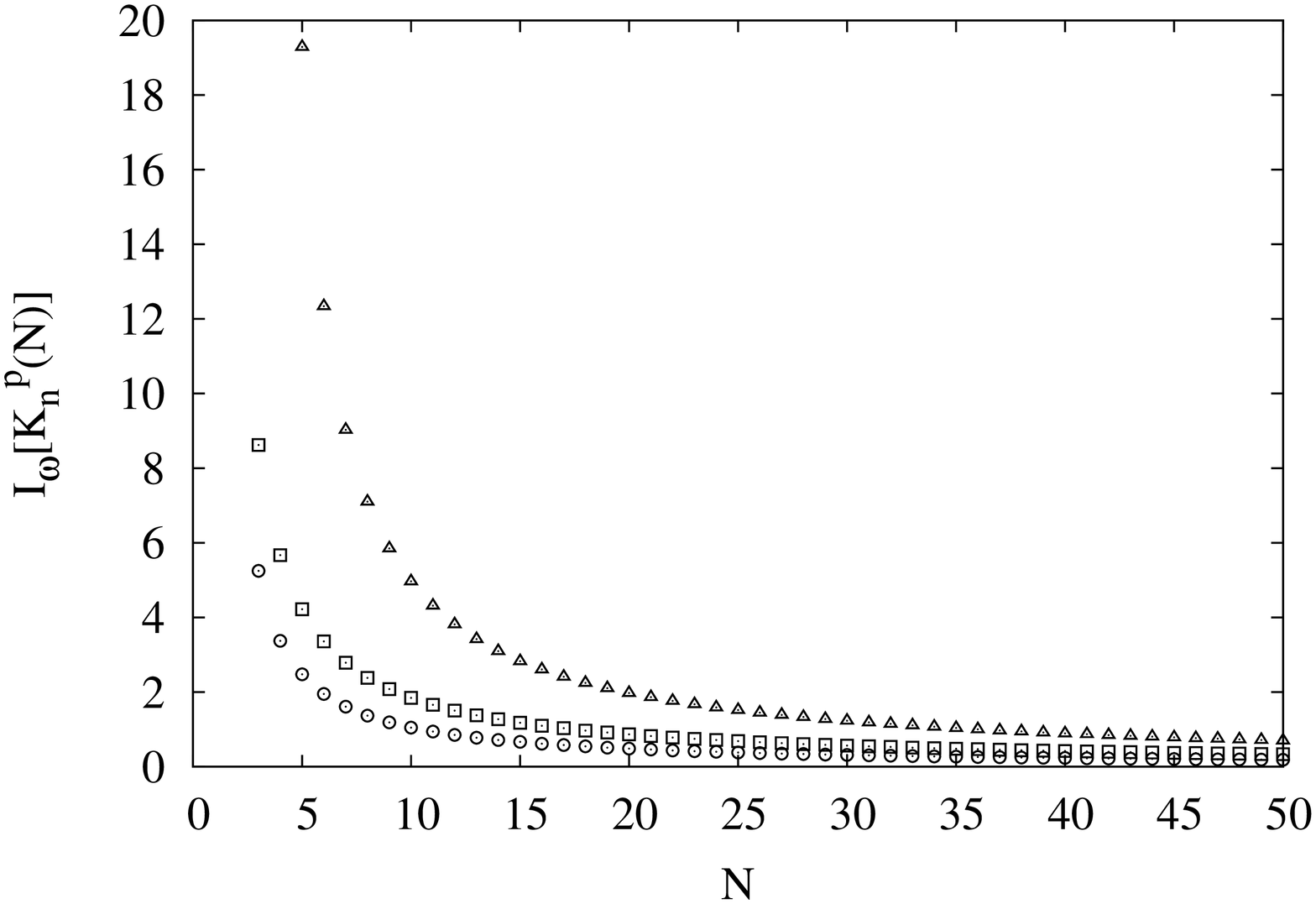}
\caption{Relative Fisher informations $I_\omega[K_2^{1/7}(N)]$ ($\square$), $I_\omega[K_4^{1/7}(N)]$ ($\vartriangle$) and $I_\omega[K_2^{1/3}(N)]$ ($\bigcirc$) as a function of the parameter $N$.}
\label{fig:kra3}
\end{center}
\end{figure}

\begin{figure}
\begin{center}
\includegraphics[height=11cm,angle=270]{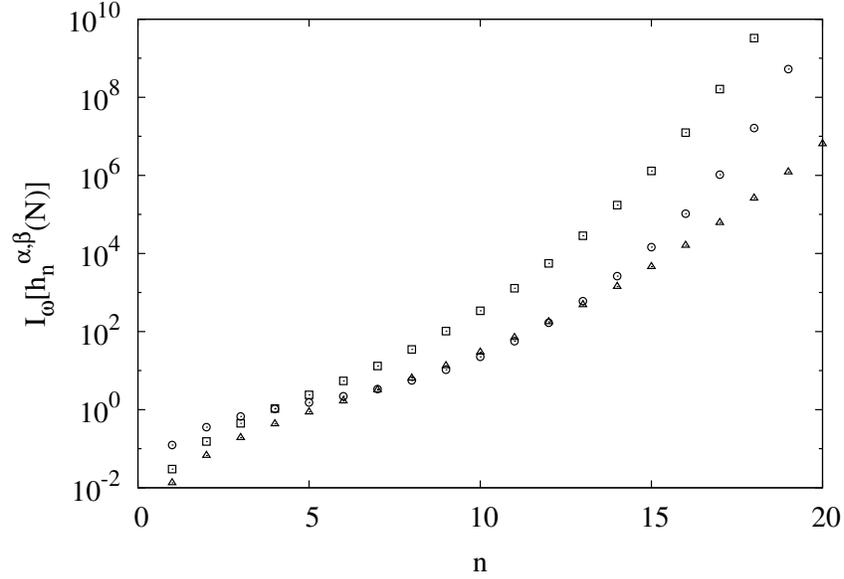}
\caption{Relative Fisher informations $I_\omega[h_n^{0,0}(20)]$ ($\square$), $I_\omega[h_n^{0,0}(30)]$ ($\vartriangle$) and $I_\omega[h_n^{3,-1/2}(20)]$ ($\bigcirc$) as a function of the degree $n$.}
\label{fig:hahn1}
\end{center}
\end{figure}

\begin{figure}
\begin{center}
\includegraphics[height=11cm,angle=270]{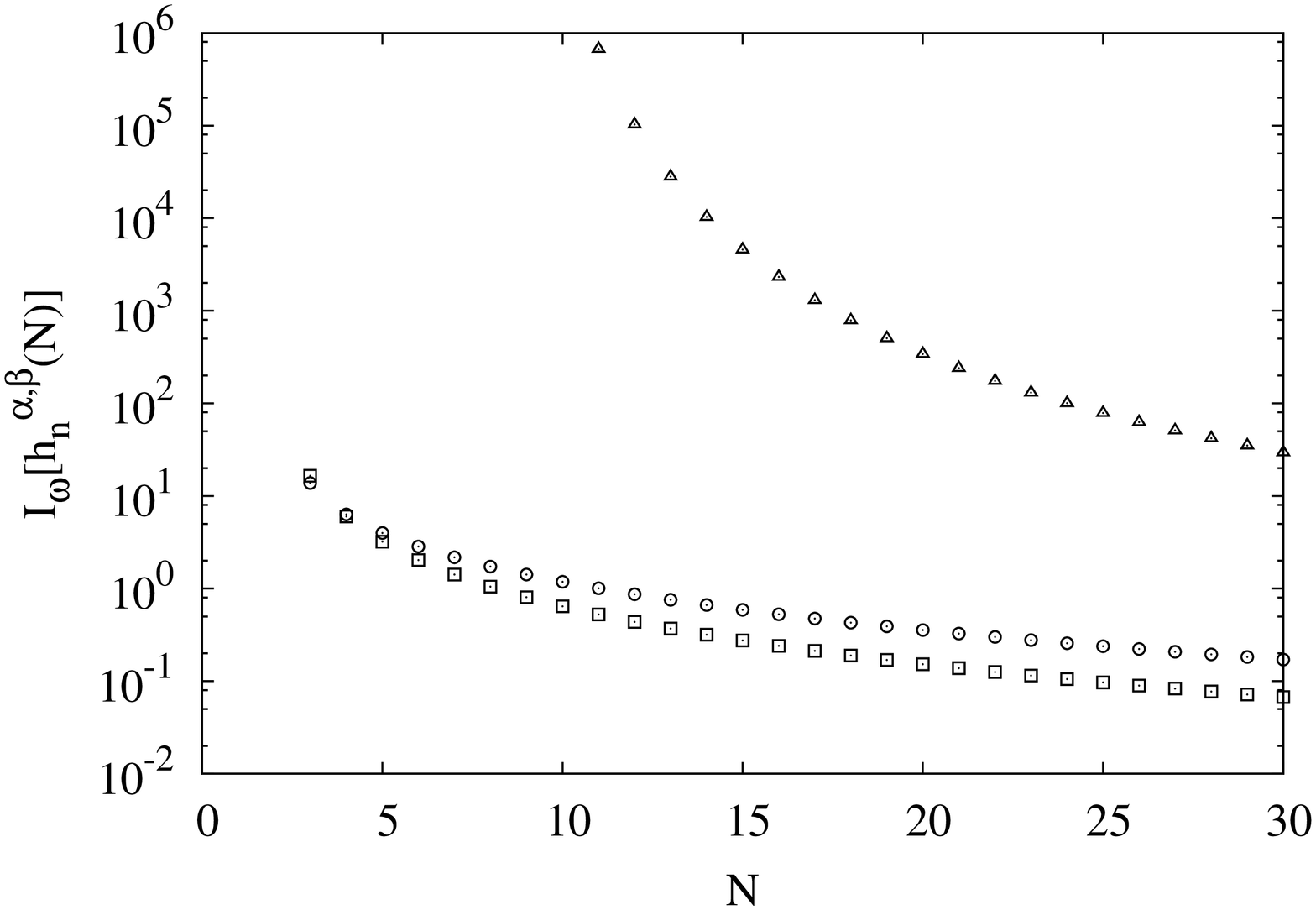}
\caption{Relative Fisher informations $I_\omega[h_2^{0,0}(N)]$ ($\square$), $I_\omega[h_10^{0,0}(N)]$ ($\vartriangle$) and $I_\omega[h_2^{3,-1/2}(N)]$ ($\bigcirc$) as a function of the parameter $N$.}
\label{fig:hahn2}
\end{center}
\end{figure}

\begin{figure}
\begin{center}
\includegraphics[height=11cm,angle=270]{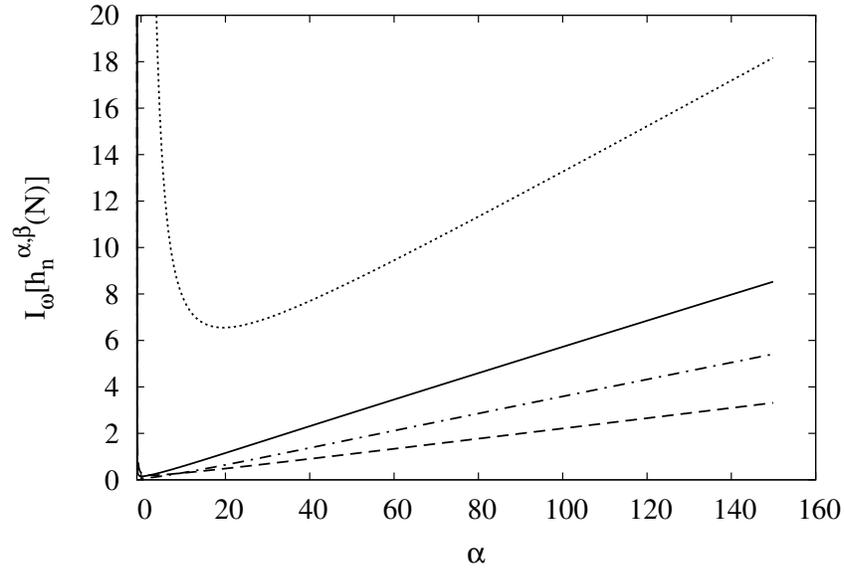}
\caption{Relative Fisher informations $I_\omega[h_2^{\alpha,0}(20)]$ (solid line), $I_\omega[h_2^{\alpha,3}(20)]$ (dashed line), $I_\omega[h_10^{\alpha,0}(20)]$ (dotted) and $I_\omega[h_2^{\alpha,0}(30)]$ (dash-dotted line) as a function of the parameter $\alpha$.}
\label{fig:hahn3}
\end{center}
\end{figure}

\begin{figure}
\begin{center}
\includegraphics[height=11cm,angle=270]{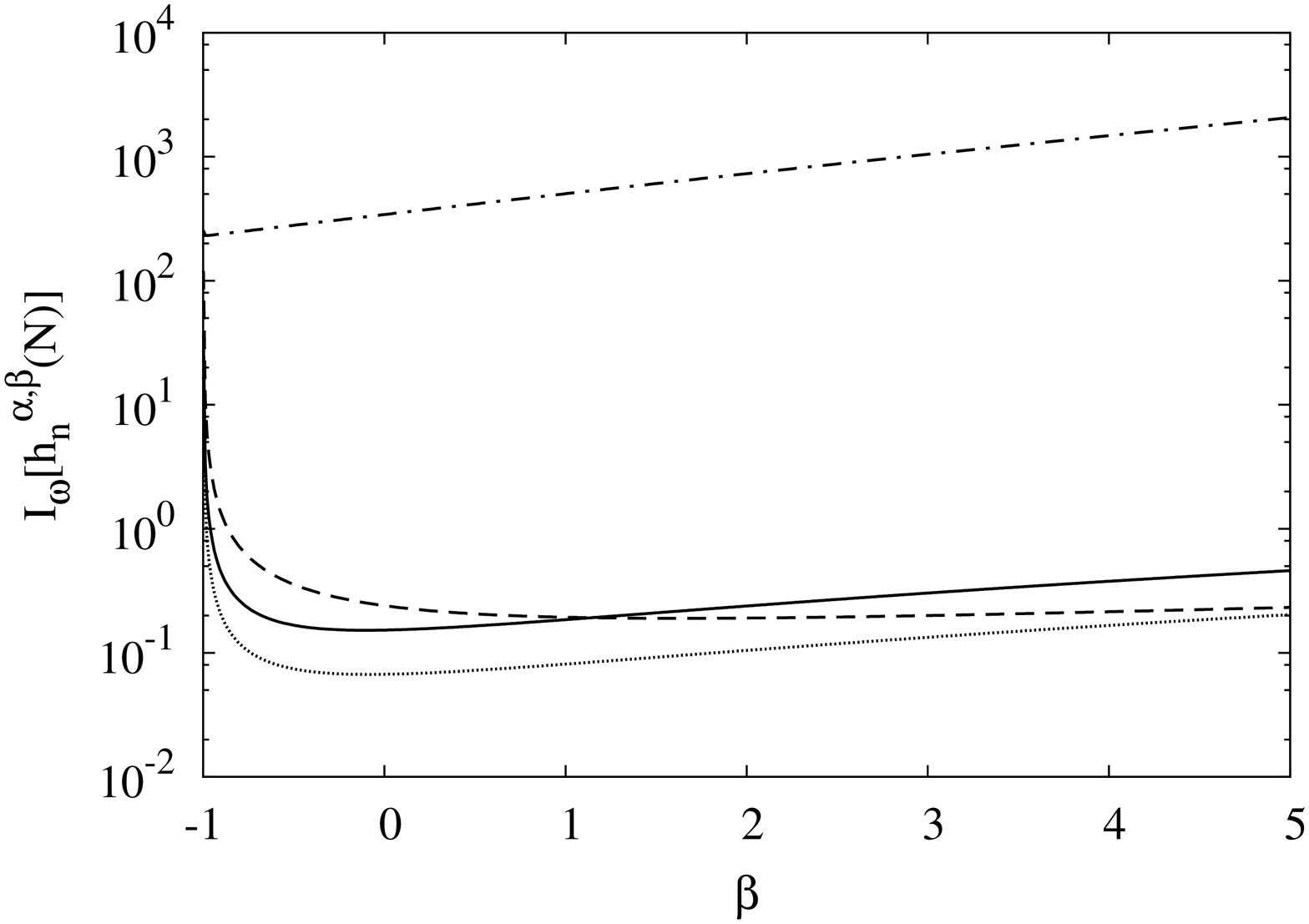}
\caption{Relative Fisher informations $I_\omega[h_2^{0,\beta}(20)]$ (solid line), $I_\omega[h_2^{3,\beta}(20)]$ (dashed line), $I_\omega[h_{10}^{0,\beta}(20)]$ (dotted line) and $I_\omega[h_2^{0,\beta}(30)]$ (dash-dotted line) as a function of the parameter $\beta$.}
\label{fig:hahn4}
\end{center}
\end{figure}


\begin{thebibliography}{10}

\bibitem{alvareznodarse_03}
R.~\'{A}lvarez Nodarse, \emph{Polinomios hipergeom\'etricos y $q$-polinomios},
  Universidad de Zaragoza, 2003.

\bibitem{alvareznodarse:etna07}
R.~\'Alvarez-Nodarse, N.~M. Atakishiyev, and R.~S. Costas-Santos,
  \emph{Factorization of the hypergeometric-type difference equation on the
  uniform lattice}, Elec. Trans. Num. Anal. \textbf{27} (2007), 34--50.

\bibitem{alvareznodarse:amc02}
R.~\'Alvarez-Nodarse and J.~S. Dehesa, \emph{Distribution of zeros of discrete
  and continuous polynomials from their recurrence relation}, Appl. Math.
  Comput. \textbf{128} (2002), 167--190.

\bibitem{alvareznodarse:jcam97}
R.~\'Alvarez-Nodarse, R.~J. Y\'a{\~n}ez, and J.~S. Dehesa, \emph{Modified
  {C}lebsh-{G}ordan-type expansions for products of discrete hypereometric
  polynomials}, J. Comput. Appl. Math. \textbf{89} (1997), 171--197.

\bibitem{aptekarev:rassm95}
A.~I. Aptekarev, V.~S. Buyarov, and J.~S. Dehesa, \emph{Asymptotic behavior of
  the {$L^p$}-norms and the entropy for general orthogonal polynomials},
  Russian Acad. Sci. Sb. Math. \textbf{82} (1995), 373--395.

\bibitem{aptekarev:ca08}
A.~I. Aptekarev, J.~S. Dehesa, A.~Mart\'{\i}nez-Finkelshtein, and R.~J.
  Y\'a{\~n}ez, \emph{Discrete entropies of orthogonal polynomials},
  Constructive Approximation \textbf{30} (2009), 93--119.

\bibitem{aptekarev:jcam09}
A.~I. Aptekarev, A.~Mart\'{\i}nez-Finkelshtein, and J.~S. Dehesa,
  \emph{Asymptotics of orthogonal polynomials entropy}, J. Comput. Appl. Math. \textbf{233}
  (2009), 1355--1365.

\bibitem{atakishiyev:tmp91}
N.~M. Atakishiyev and S.~K. Suslov, \emph{Difference analogs of the harmonic
  oscillator}, Theor. Math. Phys. \textbf{85} (1991), 442--444.

\bibitem{boykin:ejp05}
T.~B. Boykin and G.~Klimeck, \emph{The discretized {S}chr\"odinger equation for
  the finite square well and its relationship to solid-state physics}, Eur. J.
  Phys. \textbf{26} (2005), 865--881.

\bibitem{buyarov:jat99}
V.~S. Buyarov, J.~S. Dehesa, A.~Mart\'{\i}nez-Finkelshtein, and E.~B. Saff,
  \emph{Asymptotics of the information entropy for {J}acobi and {L}aguerre
  polynomials with varying weights}, J. Approx. Theory \textbf{99} (1999),
  153--166.

\bibitem{buyarov:sjsc04}
V.~S. Buyarov, J.~S. Dehesa, A.~Mart\'{\i}nez-Finkelshtein, and
  J.~S\'anchez-Lara, \emph{Computation of the entropy of polynomials orthogonal
  on an interval}, SIAM J. Sci. Comput. \textbf{26} (2004), 488--509.

\bibitem{carballo:aml01}
G.~Carballo, R.~\'{A}lvarez{-}Nodarse, and J.~S. Dehesa, \emph{Chebyshev
  polynomials in a speech recognition model}, Applied Math. Letters \textbf{14}
  (2001), 581--585.

\bibitem{chihara_78}
Th.~S. Chihara, \emph{An introduction to orthogonal polynomials}, Gordon and
  Breach, New York, 1978.

\bibitem{defazio:ijqc03}
D.~de~Fazio, S.~Cavalli, and V.~Aquilanti, \emph{Orthogonal polynomials of a
  discrete variable as expansion basis sets in quantum mechanics:
  hyperquantization algorithm}, Int. J. Quant. Chem. \textbf{93} (2003),
  91--111.

\bibitem{dehesa_05}
J.~S. Dehesa, R.~\'Alvarez-Nodarse, P.~S\'anchez-Moreno, and R.~J. Y\'a{\~n}ez,
  \emph{Information-theoretic measures of discrete orthogonal polynomials},
  Proceed. 9th Int. Conf. on Difference Equations and Discrete Dynamical
  systems, Los Angeles, 2004 (L.~J.~S. Allen, B.~Aulbach, S.~Elaydi, and
  R.~Sacker, eds.), World Scientific, New Jersey, 2005.

\bibitem{dehesa:jmp07}
J.~S. Dehesa, S.~L\'opez-Rosa, and R.~J. Y\'a{\~n}ez,
  \emph{Information-theoretic measures of hyperspherical harmonics}, J. Math.
  Phys. \textbf{48} (2007), 043503.

\bibitem{dehesa:jcam01}
J.~S. Dehesa, A.~Mart\'{\i}nez-Finkelshtein, and J.~S\'anchez-Ruiz,
  \emph{Quantum information entropies and orthogonal polynomials}, J. Comput.
  Appl. Math. \textbf{133} (2001), 23--46.

\bibitem{dehesa:ijbc02}
J.~S. Dehesa, A.~Mart\'inez-Finkelshtein, and V.~Sorokin, \emph{Short-wave
  asymptotics of the information entropy of a circular membrane}, Int. J.
  Bifurcation and Chaos \textbf{12} (2002), 2387--2392.

\bibitem{dehesa:jmp03}
J.~S. Dehesa, A.~Mart\'inez-Finkelshtein, and V.~Sorokin, \emph{Asymptotics of information entropies of some {T}oda-like
  potentials}, J. Math. Phys. \textbf{44} (2003), 36.

\bibitem{dehesa:jcam07}
J.~S. Dehesa, B.~Olmos, and R.~J. Y\'a{\~n}ez, \emph{Parameter-based {F}isher's
  information of orthogonal polynomials}, J. Comput. Appl. Math. \textbf{214}
  (2007), 136--147.

\bibitem{dehesa:jcam06}
J.~S. Dehesa, P.~S\'anchez-Moreno, and R.~J. Y\'a{\~n}ez, \emph{Cramer-{R}ao
  information plane of orthogonal hypergeometric polynomials}, J. Comput. Appl.
  Math. \textbf{186} (2006), 523--541.

\bibitem{dehesa:maa97}
J.~S. Dehesa, W.~van Assche, and R.~J. Y\'a{\~n}ez, \emph{Information entropy
  of classical orthogonal polynomials and their application to the harmonic
  oscillator and {C}oulomb potentials}, Methods Appl. Anal. \textbf{4} (1997),
  91--110.

\bibitem{deuflhard:icse89}
P.~Deuflhard and M.~Wulkow, \emph{Computational treatment of polyreaction
  kinetics by orthogonal polynomials in a discrete variable}, Impact of
  Computing in Science and Engineering \textbf{1} (1989), 269--301.

\bibitem{doliwa:jpa07}
A.~Doliwa, R.~Korhonen, and S.~Lafortune, J. Phys. A: Math. Gen. (2007),
  Special Issue on Symmetries and Integrability of Difference Equations.

\bibitem{dominici:jcam09}
D.~Dominici, \emph{Fisher information of orthogonal polynomials I}, J. Comput.
  Appl. Math. \textbf{233} (2010), 1511--1518.

\bibitem{frieden_04}
B.~R. Frieden, \emph{Science from {F}isher {I}nformation}, Cambridge University
  Press, Cambridge, 2004.

\bibitem{garcia:jcam95}
A.~G. Garc\'{\i}a, F.~Marcell\'an, and L.~Salto, \emph{A distributional study
  of discrete classical orthogonal polynomials}, J. Comput. Appl. Math.
  \textbf{57} (1995), 147--162.

\bibitem{ismail_05}
M.~E.~H. Ismail, \emph{Classical and quantum orthogonal polynomials in one
  variable}, Encyclopedia for Mathematics and its Applications, Cambridge
  University Press, 2005.

\bibitem{jacquet_98}
P.~Jacquet and W.~Szpankowski, \emph{Entropy computations for discrete
  distributions: {T}owards an analytical information theory}, Proceed. ISIT'98,
  MIT, Cambridge, MA, 1998, p.~373.

\bibitem{jacquet:itit99}
P.~Jacquet and W.~Szpankowski, \emph{Entropy computations via analytic depoissonization}, IEEE Trans.
  Information Theory \textbf{45} (1999), 1072--1081.

\bibitem{knessl:aml98}
C.~Knessl, \emph{Integral representations and asymptotic expansions for
  {S}hannon and {R}enyi entropies}, Appl. Math. Lett. \textbf{11} (1998),
  69--74.

\bibitem{koekoek_98}
R.~Koekoek and R.~F. Swarttouw, \emph{The {A}skey-scheme of the hypergeometric
  orthogonal polynomials and its $q$-analogue}, Report n. 98-17, Faculty of
  Information Technology and Systems, Delft University of Technology, 1998,
  Electronic version in http://fa.its.tudelft.nl/\~{}koekoek/askey.

\bibitem{larsson:jat02}
L.~Larsson-Cohn, \emph{${L}^p$-norms and information entropies of {C}harlier
  polynomials}, J. Approx. Theory \textbf{117} (2002), 152--172.

\bibitem{lee_07}
D.~W. Lee, \emph{Difference equations for discrete classical multiple
  orthogonal polynomials}, Proceed. 14th Int. Conf. on Difference Equations and
  Applications, Kyoto, 2006, 2007.

\bibitem{lorente:pla01}
M.~Lorente, \emph{Continuous vs. discrete models for the quantum harmonic
  oscillator and the hydrogen atom}, Phys. Lett. A \textbf{285} (2001),
  119--126.

\bibitem{lorente:jpa01}
M.~Lorente, \emph{Raising and lowering operators, factorization and
  differential-difference operators of hypergeometric type}, J. Phys. A: Math.
  Gen. \textbf{34} (2001), 569--588.

\bibitem{lorente:jcam03}
M.~Lorente, \emph{Integrable systems on the lattice and orthogonal polynomials of
  discrete variable}, J. Comput. Appl. Math. \textbf{153} (2003), 321--330.

\bibitem{meiler:08}
M.~Meiler, R.~Cordero-Soto, and S.~K. Suslov, \emph{Solution of the {C}auchy
  problem for a time-dependent {S}chrödinger equation}, J. Math. Phys. \textbf{49} (2008), 072102.

\bibitem{mickens:jdea05}
R.~E. Mickens, \emph{A note on a discrete model for the harmonic oscillator
  {S}chr\"odinger equation in n-cartesian coordinates}, J. Diff. Eq. Appl.
  \textbf{11} (2005), 779--782.

\bibitem{nikiforov_91}
A.~F. Nikiforov, S.~K. Suslov, and V.~B. Uvarov, \emph{Classical orthogonal
  polynomials of a discrete variable}, Springer Verlag, Berlin, 1991.

\bibitem{nikiforov_88}
A.~F. Nikiforov and V.~B. Uvarov, \emph{Special functions in mathematical
  physics}, Birk\"auser-Verlag, Basel, 1988.

\bibitem{odake:jmp08}
S.~Odake and R.~Sasaki, \emph{Orthogonal polynomials from {H}ermitian
  matrices}, J. Math. Phys. \textbf{49} (2008), 053503.

\bibitem{prudnikov_86}
A. P. Prudnikov, Yu. A. Brychkov and O. I. Marichev, \emph{Integrals and Series, Volume 3}, Gordon and Breach Science Publishers, Amsterdam, 1986.

\bibitem{rakhmanov:mus77}
E.~A. Rakhmanov, \emph{On the asymptotics of the ratio of orthogonal
  polynomials}, Math. USSR Sb. \textbf{32} (1977), 199--213.

\bibitem{sanchezmoreno:jpa05}
P.~S\'anchez{-}Moreno, R.~J. Y\'a{\~n}ez, and V.~Buyarov, \emph{Asymptotics of
  the information entropy of the {A}iry function}, J. Phys. A: Math. Gen.
  \textbf{38} (2005), 9969--9978.

\bibitem{sanchezmoreno_08}
P.~S\'anchez-Moreno, R.~J. Y\'a{\~n}ez, and J.~S. Dehesa, \emph{Discrete
  densities and fisher information}, Proceed. 14th Int. Conf. on Difference
  Equations and Applications, Istanbul, 2008 (M.~Bohner, Z.~Dosla, G.~Ladas,
  M.~Unal, and A.~Zafer, eds.), 2009, 291--298.

\bibitem{sanchezruiz:jcam05}
J.~S\'anchez-Ruiz and J.~S. Dehesa, \emph{Fisher information of orthogonal
  hypergeometric polynomials}, J. Comput. Appl. Math. \textbf{182} (2005),
  150--164.

\bibitem{savva:itsf00}
V.~A. Savva, V.~I. Zelenkov, and A.~S. Mazurenko, \emph{Orthogonal polynomials
  in analytical methods of solving differential equations decribing dynamics of
  multilevel systems}, Int. Trans. Spec. Funct. \textbf{10} (2000), 299--308.

\bibitem{smirnov:jpa84}
Yu.~F. Smirnov, S.~K. Suslov, and A.~M. Shikorov, \emph{{C}lebsch-{G}ordan
  coefficients and {R}acah coefficients for the {SU(2)} and {SU(1)} groups as
  the discrete analogues of the {P}\"oschl-{T}eller potentials wavefunctions},
  J. Phys. A: Math. Gen. \textbf{17} (1984), 2157--2175.

\bibitem{suslov:sjnp84}
S.~K. Suslov, \emph{The {H}ahn polynomials in the {C}oulomb problem}, Sov. J.
  Nucl. Phys. \textbf{40} (1984), 79--82.

\bibitem{vanassche:jmp95}
W.~van Assche, R.~J. Y\'a{\~n}ez, and J.~S. Dehesa, \emph{Entropy of orthogonal
  polynomials with {F}reud weights and information entropies of the harmonic
  oscillator potential}, J. Math. Phys. \textbf{36} (1995), 4106--4118.

\bibitem{vercin:jmp98}
A.~Vercin, \emph{Ordered products, {W}-algebra and two-variable, definite
  parity, orthogonal polynomials}, J. Math. Phys. \textbf{39} (1998),
  2418--2427.

\bibitem{yanez:jmp99}
R.~J. Y\'a{\~n}ez, W.~Van Assche, R.~Gonz\'alez-F\'erez, and J.~S. Dehesa,
  \emph{Entropic integrals of hyperspherical harmonics and spatial entropy of
  {D}-dimensional central potentials}, J. Math. Phys. \textbf{40} (1999),
  5675--5686.

\bibitem{yanez:jmp08}
R.~J. Y\'a{\~n}ez, P.~S\'anchez-Moreno, A.~Zarzo, and J.~S. Dehesa,
  \emph{Fisher information of special functions and second-order differential
  equations}, J. Math. Phys. \textbf{49} (2008), 082104.

\end{thebibliography}
\end{document}